\documentclass[12pt]{amsart}

\usepackage{amsthm}
\usepackage{amssymb}
\usepackage{amsmath}

\numberwithin{equation}{section}

\newtheorem{thm}{Theorem}[section]
\newtheorem{prop}[thm]{Proposition}
\newtheorem{lem}[thm]{Lemma}

\theoremstyle{definition}

\newtheorem{defn}[thm]{Definition}

\begin{document}

\title[Deformed affine Hecke algebra and integrable stochastic system]
{A deformation of affine Hecke algebra and integrable stochastic particle system}
\author{Yoshihiro Takeyama}
\address{Division of Mathematics, 
Faculty of Pure and Applied Sciences, 
University of Tsukuba, Tsukuba, Ibaraki 305-8571, Japan}
\email{takeyama@math.tsukuba.ac.jp}

\begin{abstract}
We introduce a deformation of the affine Hecke algebra of type $GL$ 
which describes the commutation relations of the divided difference operators 
found by Lascoux and Sch\"utzenberger and the multiplication operators.  
Making use of its representation we construct an integrable stochastic particle system. 
It is a generalization of the $q$-Boson system due to Sasamoto and Wadati. 
We also construct eigenfunctions of its generator using the propagation operator.  
As a result we get the same eigenfunctions for the $(q, \mu, \nu)$-Boson process 
obtained by Povolotsky. 
\end{abstract}
%%%%%%%%%%%%%%%%%%%%%%%%%%%%%%%%%%%%%
\maketitle

\setcounter{section}{0}
\setcounter{equation}{0}

%%%%%%%%%%%%%%%%%%%%%%%%%%%%%%%%%%%%%%%%%%%%%%%%%%%%%%%%%%%%%%%%%%%%%%%%%%%%%%%

\section{Introduction}

In this article we introduce a deformation of the affine Hecke algebra of type $GL$ 
and construct an integrable stochastic particle system making use of its representation. 

In a previous paper \cite{T} we constructed a discrete analogue of 
the non-ideal Bose gas with delta-potential interactions on a circle, 
which we call the periodic delta Bose gas for short.  
While a discretization of the periodic delta Bose gas and its generalization 
was studied by van Diejen \cite{vD1, vD2} from the viewpoint of the theory of 
Macdonald's spherical functions, 
our discretization is motivated by the desire to understand an algebraic structure of 
integrable stochastic models. 

The discrete model constructed in \cite{T} contains two parameters.  
Specializing the parameters suitably and taking the limit as the system size goes to infinity, 
its Hamiltonian $H$ becomes the time evolution operator for the joint moment 
of the integrable stochastic system called the O'Connell-Yor semi-discrete directed polymer \cite{OY}.  
We can construct eigenfunctions of $H$ using the propagation operator $G$, 
which sends an eigenfunction of (a half of) the discrete Laplacian to that of $H$. 
To define $G$ we generalize the construction by by van Diejen and Emsiz \cite{DE} 
of the integral-reflection operators due to Yang \cite{Y} and Gutkin \cite{G}. 
Then the discrete integral-reflection operators determine a representation of 
the affine Hecke algebra of type $GL$. 

Recently the author found that the Hamiltonian $H$ is related to another stochastic system more closely. 
The operator $H$ acts on the space of functions on the orthogonal lattice $\mathbb{Z}^{k}$, where 
$k$ is the number of particles, 
and leaves the space of symmetric functions invariant. 
Identify the space of symmetric functions with the space of functions on the fundamental chamber 
$\mathbb{W}^{k}=\{(m_{1}, \ldots , m_{k})\in\mathbb{Z}^{k}\, | \, m_{1}\ge \cdots \ge m_{k}\}$. 
We assign to each element $(m_{1}, \ldots , m_{k})$ of $\mathbb{W}^{k}$ 
the configuration of $k$ bosonic particles on $\mathbb{Z}$ such that 
the particles are on the sites  $m_{1}, \ldots , m_{k}$. 
Then, by specializing the two parameters of $H$ in another way and adding a constant, 
we obtain the transition rate matrix of the $q$-Boson system 
introduced by Sasamoto and Wadati \cite{SW}. 

In this paper we generalize the above construction of an integrable stochastic particle system. 
Our ingredient is a deformation of the affine Hecke algebra of type $GL$. 
In \cite{LS} Lascoux and Sch\"utzenberger characterize the difference operators 
acting on polynomials which satisfy the braid relations. 
The operators contain four parameters and at a spacial point 
they turn into the Demazure-Lustzig operators 
which give a polynomial representation of the affine Hecke algebra. 
Our deformed algebra arises from the commutation relations between 
the difference operators due to Lascoux and Sch\"utzenberger and 
the multiplication operators. 
By definition it has a polynomial representation. 
Making use of it we can construct the discrete integral-reflection operators with more parameters 
and define the propagation operator $G$ as before. 
One of the main results of this article is construction of the discrete Hamiltonian $H$ 
satisfying the commutation relation 
$HG=G\Delta$, where $\Delta$ is the discrete Laplacian, in this generalized setting. 

Our operator $H$ also leaves the space of symmetric functions invariant, 
and by specializing the four parameters suitably we obtain 
a transition rate matrix of a continuous time Markov chain on $\mathbb{W}^{k}$. 
The resulting model is described as follows. 
It is a stochastic particle system on the one-dimensional lattice $\mathbb{Z}$ 
controlled by two parameters $s$ and $q$. 
The particles can occupy the same site simultaneously. 
Some particles may move from site $i$ to $i-1$ 
independently for each $i \in \mathbb{Z}$. 
The rate at which $r$ particles move to the left from a cluster with $c$ particles 
is given by 
\begin{align*}
\frac{s^{r-1}}{[r]}\prod_{p=0}^{r-1}\frac{[c-p]}{1+s[c-1-p]} \qquad (c \ge r \ge 1),   
\end{align*}
where $[n]:=(1-q^{n})/(1-q)$ is the $q$-integer. 
In the case of $s=0$, the rate is equal to zero unless $r=1$ and hence 
only one particle may move with the rate proportional to $1-q^{c}$. 
Thus we recover the $q$-Boson system. 

Using the propagation operator $G$ 
we can construct symmetric eigenfunctions of the operator $H$ 
by means of the Bethe ansatz method, which we call the Bethe wave functions.   
After the specialization of the parameters we obtain the eigenfunctions of the transition rate matrix. 
They are parameterized by a tuple $z=(z_{1}, \ldots , z_{k})$ of distinct constants 
and are given by 
\begin{align*}
\sum_{\sigma \in \mathfrak{S}_{k}}\prod_{1\le i<j \le k}
\frac{qz_{\sigma(i)}-z_{\sigma(j)}}{z_{\sigma(i)}-z_{\sigma(j)}}
\prod_{i=1}^{k}\left(\frac{1-\nu z_{\sigma(i)}}{1-z_{\sigma(j)}}\right)^{\!m_{i}} 
\qquad 
((m_{1}, \ldots , m_{k}) \in \mathbb{W}^{k}),    
\end{align*}
where $\nu:=s/(1-q+s)$. 
Here we note that they are equal to the eigenfunctions of the generator of 
the $(q, \mu, \nu)$-Boson process introduced by Povolotsky \cite{P} (see also \cite{C}).  

The paper is organized as follows. 
In Section \ref{sec:2} we define the deformation of the affine Hecke algebra 
and introduce its representations which are the origin of the propagation operator. 
In Section \ref{sec:hamiltonian} we define the discrete Hamiltonian $H$ and the propagation operator $G$,  
and prove the commutation relation $HG=G\Delta$. 
Using the operator $G$ we construct the Bethe wave functions. 
In Section \ref{sec:4} we describe the construction of the stochastic particle system 
arising from the Hamiltonian $H$. 
We prove some polynomial identities which we use to rewrite the operator $H$ in Appendix \ref{sec:app}. 

%%%%%%%%%%%%%%%%%%%%%%%%%%%%%%%%%%%%%%%%%%%%%%%%%%%%%%%%%%%%%%%%%%%%%%%%%%%%%%%

\section{A Deformation of Affine Hecke Algebra and Its Representation}\label{sec:2}

\subsection{Preliminaries}

Throughout this paper we fix an integer $k \ge 2$. 
Let $V:=\oplus_{i=1}^{k}\mathbb{R}v_{i}$ be the $k$-dimensional Euclidean space  
with an orthonormal basis $\{v_{i}\}_{i=1}^{k}$,  
and $V^{*}$ the linear dual of $V$. 
We let $\{\epsilon_{i}\}_{i=1}^{k}$ denote the dual basis of $V^{*}$ 
corresponding to $\{v_{i}\}_{i=1}^{k}$. 
Set $\alpha_{ij}:=\epsilon_{i}-\epsilon_{j}$ for $i, j=1, \ldots , k$. 
The subset $R:=\{\alpha_{ij} \, | \, i\not=j\}$ of $V^{*}$ forms 
the root system of type $A_{k-1}$ with the simple roots $a_{i}:=\alpha_{i, i+1} \, (1\le i<k)$. 
Denote by $R^{\pm}$ the set of the associated positive and negative roots. 
For $v \in V$, set 
\begin{align*}
I(v):=\{a \in R^{+}\, | \, a(v)<0\}.  
\end{align*}

The Weyl group $W$ of type $A_{k-1}$ is generated by 
the orthogonal reflections $s_{i} \, : \, V \to V \, (1\le i<k)$ defined by 
$s_{i}(v):=v-a_{i}(v)a^{\vee}_{i}$, where $a^{\vee}_{i}:=v_{i}-v_{i+1}$ is the simple coroot. 
Denote the length of $w \in W$ by $\ell(w)$. 

For any $v \in V$, the orbit $Wv$ intersects the closure of the fundamental chamber 
\begin{align*}
\overline{C_{+}}:=\{v \in V \, | \, a_{i}(v)\ge 0 \,\, (i=1, \ldots , k-1)\}  
\end{align*}
at one point. 
Take a shortest element $w \in W$ such that $wv \in \overline{C_{+}}$. 
Then $I(v)=R^{+}\cap w^{-1}R^{-}$ and hence the shortest element is uniquely determined. 
Denote it by $w_{v}$. 

We will make use of the following proposition. 

\begin{prop}\label{prop:key}
Suppose that $v, v' \in V$ satisfy $I(v) \subset I(v')$. 
Then $w_{v'}=w_{w_{v}v'}w_{v}$ and 
$\ell(w_{v'})=\ell(w_{w_{v}v'})+\ell(w_{v})$. 
\end{prop}

\subsection{A deformation of affine Hecke algebra}

Let us define a deformation of the affine Hecke algebra of type $GL_{k}$. 

\begin{defn}
Let $\alpha, \beta, \gamma, \delta$ be complex constants and set 
\begin{align*}
q:=1+\beta\gamma-\alpha\delta.   
\end{align*}
We define the algebra $\mathcal{A}_{k}$ to be the unital associative $\mathbb{C}$-algebra with 
the generators $X_{i}^{\pm 1} \, (1\le i\le k)$ and $T_{i} \, (1\le i<k)$ 
satisfying the following relations: 
\begin{align*}
& 
(T_{i}-1)(T_{i}+q)=0 \quad (1 \le i<k), \qquad  
T_{i}T_{i+1}T_{i}=T_{i}T_{i+1}T_{i} \quad (1 \le i \le k-2), \\ 
& 
T_{i}T_{j}=T_{j}T_{i} \quad (|i-j|>1), \quad 
X_{i}X_{j}=X_{j}X_{i} \quad (i, j=1, \ldots , k), \\ 
& 
X_{i+1}T_{i}-T_{i}X_{i}=T_{i}X_{i+1}-X_{i}T_{i}=(\alpha+\beta X_{i})(\gamma+\delta X_{i+1}) \quad 
(1\le i<k), \\ 
& 
X_{i}T_{j}=T_{j}X_{i}\quad (i\not=j, j+1). 
\end{align*}
\end{defn}

When $\beta=\gamma=0$, the algebra $\mathcal{A}_{k}$ is isomorphic to 
the affine Hecke algebra of type $GL_{k}$.  
We will use the property that 
any symmetric polynomial in $X_{1}, \ldots , X_{k}$ commutes with 
$T_{i} \, (1\le i<k)$ in $\mathcal{A}_{k}$.  

Set 
\begin{align*}
L:=\bigoplus_{i=1}^{k}\mathbb{Z}v_{i}  
\end{align*}
and denote by $F(L)$ the vector space of $\mathbb{C}$-valued functions on $L$. 
The Weyl group acts on $F(L)$ by $(wf)(x):=f(w^{-1}x)$. 
Set 
\begin{align*}
F(L)^{W}:=\{f \in F(L) \, | \, wf=f \, \hbox{for all} \, w \in W\}.  
\end{align*}
 
Now we introduce a right action of $\mathcal{A}_{k}$ 
on the group algebra $\mathbb{C}[L]$ due to Lascoux and Sch\"utzenberger \cite{LS}. 
In the following we identify $\mathbb{C}[L]$ with the Laurent polynomial ring 
$\mathbb{C}[e^{\pm v_{1}}, \ldots , e^{\pm v_{k}}]$. 

\begin{prop}\label{prop:right-action}\cite{LS}
Define the $\mathbb{C}$-linear operators $\check{X}_{i} \, (1\le i \le k)$ and 
$\check{T}_{i} \, (1\le i<k)$ acting on $\mathbb{C}[L]$ from the right by 
\begin{align*}
P\check{X}_{i}:=e^{-v_{i}}P, \qquad 
P\check{T}_{i}:=P.s_{i}+\frac{(\alpha e^{v_{i}}+\beta)(\gamma e^{v_{i+1}}+\delta)}{e^{v_{i}}-e^{v_{i+1}}}
\left(P-P.s_{i}\right),  
\end{align*}
where $.$ stands for the right action of the Weyl group defined by 
$e^{x}.w:=e^{w^{-1}(x)} \, (x \in L, w \in W)$. 
Then the assignment $X_{i} \mapsto \check{X}_{i}$ and $T_{i} \mapsto \check{T}_{i}$ extends uniquely to 
a right representation of the algebra $\mathcal{A}_{k}$ on $\mathbb{C}[L]$. 
\end{prop}

Consider the non-degenerate bilinear pairing $\mathbb{C}[L]\times F(L) \to \mathbb{C}$ defined by 
$(e^{x}, f):=f(x)$ for $x \in L$ and $f \in F(L)$. 
We define the $\mathbb{C}$-linear operators $\widehat{X}_{i} \, (1\le i \le k)$ and 
$\widehat{T}_{i} \, (1\le i<k)$ acting on $F(L)$ by 
\begin{align*}
(P\check{X_{i}}, f)=(P, \widehat{X}_{i}f), \quad 
(P\check{T_{i}}, f)=(P, \widehat{T}_{i}f). 
\end{align*}
{}From Proposition \ref{prop:right-action} they give a left action of $\mathcal{A}_{k}$ on $F(L)$:  

\begin{prop}
The assignment $X_{i} \mapsto \widehat{X}_{i}$ and $T_{i} \mapsto \widehat{T}_{i}$ extends uniquely to 
a left representation of the algebra $\mathcal{A}_{k}$ on $F(L)$. 
The action is explicitly given as follows: 
\begin{align*}
(\widehat{X}_{i}f)(x)=f(x-v_{i}).  
\end{align*}  
If $a_{i}(x)>0$ then 
\begin{align*}
(\widehat{T}_{i}f)(x)&= 
\alpha\delta f(x)+(1+\beta\gamma)f(s_{i}x)+
\alpha\gamma\sum_{j=1}^{a_{i}(x)}f(s_{i}x+ja^{\vee}_{i}+v_{i+1}) \\ 
&+(\alpha\delta+\beta\gamma)\sum_{j=1}^{a_{i}(x)-1}f(s_{i}x+ja^{\vee}_{i})+
\beta\delta\sum_{j=0}^{a_{i}(x)-1}f(s_{i}x+ja^{\vee}_{i}-v_{i+1}). 
\end{align*}
When $a_{i}(x)=0$, we have $(\widehat{T}_{i}f)(x)=f(x)$. 
If $a_{i}(x)<0$ then 
\begin{align*}
(\widehat{T}_{i}f)(x)&=-\beta\gamma f(x)+(1-\alpha\delta)f(s_{i}x)-
\alpha\gamma\sum_{j=0}^{-a_{i}(x)-1}f(s_{i}x-ja^{\vee}_{i}+v_{i+1}) \\ 
&-(\alpha\delta+\beta\gamma)\sum_{j=1}^{-a_{i}(x)-1}f(s_{i}x-ja^{\vee}_{i})-
\beta\delta\sum_{j=1}^{-a_{i}(x)}f(s_{i}x-ja^{\vee}_{i}-v_{i+1}). 
\end{align*}
\end{prop}

We will often use the fact that $(\widehat{T}_{i}f)(x)=0$ for any $f \in F(L)$ if $a_{i}(x)=0$.

%%%%%%%%%%%%%%%%%%%%%%%%%%%%%%%%%%%%%%%%%%%%%%%%%%%%%%%%%%%%%%%%%%%%%%%%%%%%%%%

\section{Discrete Hamiltonian and Propagation Operator}\label{sec:hamiltonian}

\subsection{Discrete Hamiltonian}

Hereafter we assume that 
\begin{align*}
1+\beta\gamma[n]\not=0 
\end{align*}
for any positive integer $n$, where 
\begin{align*}
[n]:=\frac{1-q^{n}}{1-q}
\end{align*}
is the $q$-integer. 

We define the functions 
$d^{\pm}_{i} \, (1\le i \le k)$ and 
$\delta_{j_{1}, j_{2}, \ldots , j_{r}}\, (1\le j_{1}<j_{2}<\cdots <j_{r}\le k)$ on $L$ by 
\begin{align}
d^{+}_{i}(x)&:=\#\{p \, | \, i<p\le k, \, \alpha_{ip}(x)=0 \}, 
\label{eq:def-d2} 
\\ 
d^{-}_{i}(x)&:=\#\{p \, | \, 1\le p<i, \, \alpha_{pi}(x)=0 \}. 
\nonumber 
\end{align}
and 
\begin{align*}
\delta_{j_{1}, j_{2}, \ldots , j_{r}}(x):=\left\{ 
\begin{array}{ll}
1 & (\epsilon_{j_{1}}(x)=\cdots =\epsilon_{j_{r}}(x)), \\
0 & (\hbox{otherwise}).
\end{array}
\right. 
\end{align*}
If $r=1$ we set $\delta_{j}\equiv 1$ by definition. 

Now we define the discrete Hamiltonian $H$ by 
\begin{align*}
H:&=-\alpha\gamma\sum_{j=1}^{k}\frac{[d_{j}^{+}]}{1+\beta\gamma[d_{j}^{+}]} \\ 
&+\sum_{r=1}^{k}(-\beta\delta)^{r-1}[r-1]! \, q^{-r(r-1)/2}\!\!\!\!\!\!
\sum_{1\le j_{1}<\cdots <j_{r}\le k}
\frac{q^{\sum_{p=1}^{r}d_{j_{p}}^{-}}\delta_{j_{1}, \ldots , j_{r}}}
{\prod_{p=0}^{r-1}(1+\beta\gamma[d_{j_{1}}^{+}+d_{j_{1}}^{-}-p])}
\prod_{p=1}^{r}\widehat{X}_{j_{p}},   
\end{align*}
where $[n]!:=\prod_{a=1}^{n}[a]\, (n>0)$ and $[0]!:=1$.  
Note that the index $j_{1}$ in the factor $1+\beta\gamma[d_{j_{1}}^{+}+d_{j_{1}}^{-}-p]$ 
may be replaced with any $j_{p}$ because 
$d_{i}^{+}(x)+d_{i}^{-}(x)=d_{j}^{+}(x)+d_{j}^{-}(x)$ if $\alpha_{ij}(x)=0$. 

Using the equality 
\begin{align*}
\sum_{j=1}^{k}[d_{j}^{+}]=\sum_{j=1}^{k}q^{d_{j}^{-}}d_{j}^{+},   
\end{align*}
we see that when $\beta=0$ it holds that 
\begin{align*}
H=\sum_{j=1}^{k}q^{d_{j}^{-}}(\widehat{X}_{j}-\alpha\gamma d_{j}^{+}).    
\end{align*}
This operator is introduced in \cite{T} as a discrete analogue of the Hamiltonian of the delta Bose gas 
under periodic boundary condition\footnote{
The parameters $\alpha$ and $\beta$ in \cite{T} are equal to 
$\alpha\gamma$ and $q=1-\alpha\delta$, respectively.}. 

For convenience we write down the action of $H$ more explicitly. 
For a non-empty subset $J=\{j_{1}, \ldots , j_{m}\} \, (j_{1}<\cdots <j_{m})$ of 
$\{1, 2, \ldots , k\}$, 
we define the operator $H_{J}$ acting on $F(L)$ by 
\begin{align}\label{eq:def-HJ}
H_{J}&:=-\alpha\gamma\sum_{d=1}^{m-1}\frac{[d]}{1+\beta\gamma[d]} \\ 
&\quad {}+\sum_{r=1}^{m}
\frac{(-\beta\delta)^{r-1}[r-1]!\,q^{-r(r-1)/2}}{\prod_{p=0}^{r-1}(1+\beta\gamma[m-1-p])}
e_{r}(\widehat{X}_{j_{1}}, q\widehat{X}_{j_{2}}, \ldots , q^{m-1}\widehat{X}_{j_{m}}), 
\nonumber 
\end{align}
where $e_{r}$ is the elementary symmetric polynomial of degree $r$. 
Then the value $(Hf)(x)$ is written as follows. 

\begin{lem}\label{lem:H-decompose}
For $x \in L$, 
decompose the set $\{1, 2, \ldots , k\}$ into a direct sum $\sqcup_{n=1}^{N}J_{n}^{x}$ so that 
$i$ and $j$ belong to the same subset $J_{n}^{x}$ if and only if $\alpha_{ij}(x)=0$. 
Then for any $f \in F(L)$ it holds that 
\begin{align}\label{eq:H-to-HJ}
(Hf)(x)=\sum_{n=1}^{N}(H_{J_{n}^{x}}f)(x).  
\end{align}
\end{lem}

{}From the expression \eqref{eq:H-to-HJ} we see that 

\begin{prop}
$H\left(F(L)^{W}\right) \subset F(L)^{W}$. 
\end{prop}

\begin{proof}
Suppose that $f \in F(L)^{W}$. 
It suffices to show that $(Hf)(s_{i}x)=Hf(x)$ for any $x \in L$ and $1\le i<k$.  
If $a_{i}(x)=0$ it is trivial. 
Let us consider the case of $a_{i}(x)\not=0$.  
Denote the transposition $(i, i+1) \in \mathfrak{S}_{k}$ by $\tau$. 
Consider the decomposition $\{1, 2, \ldots , k\}=\sqcup_{n=1}^{N}J_{n}^{s_{i}x}$ 
given in Lemma \ref{lem:H-decompose} with $x$ replaced by $s_{i}x$. 
Then it holds that $J_{n}^{s_{i}x}=\tau(J_{n}^{x})$. 
Note that $\{i, i+1\}\not\subset \tau(J_{n}^{x})$ for any $n$ because $a_{i}(x)\not=0$. 
For $1\le j_{1}<\cdots <j_{m}\le k$ satisfying $\{i, i+1\}\not\subset\{j_{1}, \ldots , j_{m}\}$ 
and $0 \le r \le m$, 
it holds that 
\begin{align*}
\left(
e_{r}(\widehat{X}_{\tau(j_{1})}, q\widehat{X}_{\tau(j_{2})}, \ldots , q^{m-1}\widehat{X}_{\tau(j_{m})})f
\right)(s_{i}x)=\left(
e_{r}(\widehat{X}_{j_{1}}, q\widehat{X}_{j_{2}}, \ldots , q^{m-1}\widehat{X}_{j_{m}})f
\right)(x)
\end{align*}
because $f \in F(L)^{W}$. 
Therefore $(H_{\tau(J_{n}^{x})}f)(s_{i}x)=(H_{J_{n}^{x}}f)(x)$. 
{}From Lemma \ref{lem:H-decompose} we find that $(Hf)(s_{i}x)=(Hf)(x)$. 
\end{proof}

For later use we rewrite the operator $H_{J}$ using the two equalities below. 
See Appendix \ref{sec:app} for the proof. 

\begin{lem}\label{lem:H-rewrite-1}
Let $m$ be a positive integer and $z_{1}, \ldots , z_{m}$ commutative indeterminates. 
Then the following equality holds. 
\begin{align*}
& 
\sum_{r=1}^{m}\frac{(-\beta\delta)^{r-1}[r-1]!\,q^{-r(r-1)/2}}
{\prod_{p=0}^{r-1}(1+\beta\gamma[m-1-p])}\,
e_{r}(z_{1}, qz_{2}, \ldots , q^{m-1}z_{m}) \\  
&=\frac{1}{\beta}\sum_{r=1}^{m}\frac{(-\delta)^{r-1}[r-1]!}
{\prod_{p=1}^{r}(1+\beta\gamma[p-1])}
\sum_{1\le b_{1}<\cdots <b_{r}\le m}q^{\sum_{p=1}^{r}(b_{p}-m)}
\left\{\prod_{p=1}^{r}(\alpha+\beta q^{p-1}z_{b_{p}})-\alpha^{r}\right\}. 
\end{align*}
\end{lem}

\begin{lem}\label{lem:H-rewrite-2}
Let $m$ be a positive integer. 
Then the following equality holds.
\begin{align*}
\frac{1}{\beta}\sum_{r=1}^{m}\frac{(-\delta)^{r-1}[r-1]!}{\prod_{p=1}^{r}(1+\beta\gamma[p-1])}
\,\,\alpha^{r}\!\!\!\sum_{1\le b_{1}<\cdots <b_{r}\le m}\!\!\!q^{\sum_{p=1}^{r}(b_{p}-m)}=
\frac{\alpha}{\beta}\sum_{d=0}^{m-1}\frac{1}{1+\beta\gamma[d]}
\end{align*}  
\end{lem}

Lemma \ref{lem:H-rewrite-1} and Lemma \ref{lem:H-rewrite-2} imply the following formula. 

\begin{prop}\label{prop:H-rewrite}
Let $J=\{j_{1}, \ldots , j_{m}\} \, (j_{1}<\cdots <j_{m})$ 
be a non-empty subset of $\{1, 2, \ldots , k\}$.  
Then it holds that 
\begin{align*}
H_{J}&=-\frac{\alpha}{\beta}m \\ 
&+\frac{1}{\beta}\sum_{r=1}^{m}\frac{(-\delta)^{r-1}[r-1]!}
{\prod_{p=1}^{r}(1+\beta\gamma[p-1])}
\sum_{1\le b_{1}<\cdots <b_{r}\le m}q^{\sum_{p=1}^{r}(b_{p}-m)}
\prod_{p=1}^{r}(\alpha+\beta q^{p-1}\widehat{X}_{j_{b_{p}}}). 
\end{align*}
\end{prop}

\subsection{Propagation operator}

Let $w$ be an element of the Weyl group $W$ and 
$w=s_{i_{1}}\cdots s_{i_{r}} \in W$ a reduced expression. 
Then we set $\widehat{T}_{w}:=\widehat{T}_{i_{1}}\cdots \widehat{T}_{i_{r}}$.   
It does not depend on the choice of the reduced expression of $w$. 

\begin{defn}
We define the  \textit{propagation operator} $G \, : \, F(L) \to F(L)$ by 
\begin{align*}
G(f)(x):=(\widehat{T}_{w_{x}}f)(w_{x}x).   
\end{align*}
\end{defn}

Hereafter, for $x \in L$, we denote by $\sigma_{x}$ 
the element of the symmetric group $\mathfrak{S}_{k}$ 
determined by 
\begin{align*}
w_{x}(v_{i})=v_{\sigma_{x}(i)} \qquad (1 \le i \le k).  
\end{align*}
Then $\epsilon_{i}(x)=\epsilon_{\sigma_{x}(i)}(w_{x}x)$ for $1\le i \le k$.  

In the rest of this subsection we prove the following proposition. 

\begin{prop}\label{prop:G-shift}
Suppose that $1\le t_{1}<\cdots <t_{r}\le k$ and  
that $x \in L$ satisfies $\epsilon_{t_{1}}(x)=\cdots =\epsilon_{t_{r}}(x)$. 
Then for any $f \in F(X)$ it holds that 
\begin{align*}
& 
(\widehat{X}_{t_{1}}\cdots \widehat{X}_{t_{r}}G(f))(x) \\ 
&=\left(
\widehat{X}_{l-r+1}\cdots \widehat{X}_{l}\,
(\widehat{T}_{l-r}\cdots \widehat{T}_{\sigma_{x}(t_{1})})\cdots 
(\widehat{T}_{l-1}\cdots \widehat{T}_{\sigma_{x}(t_{r})})\widehat{T}_{w_{x}}(f)
\right)(w_{x}x),    
\end{align*}
where $l=\sigma_{x}(t_{1})+d_{t_{1}}^{+}(x)$.  
\end{prop}

First we note that the functions $d_{i}^{\pm}$ have the following properties. 

\begin{lem}
(1)\, For $x \in L, 1\le i\le k$ and $1\le j<k$, it holds that 
\begin{align*}
d_{i}^{\pm}(s_{j}x)=\left\{ 
\begin{array}{ll}
d_{i}^{\pm}(x) & (j\not=i-1, i), \\ 
d_{i-1}^{\pm}(x)\mp\theta(a_{i-1}(x)=0) & (j=i-1), \\ 
d_{i+1}^{\pm}(x)\pm\theta(a_{i}(x)=0) & (j=i), 
\end{array}
\right. 
\end{align*} 
where $\theta(P)=1$ or $0$ if $P$ is true or false, respectively. 

(2)\, For any $x \in L$ and $1 \le i \le k$, 
it holds that $d_{i}^{\pm}(x)=d_{\sigma_{x}(i)}^{\pm}(w_{x}x)$.  
\end{lem}

\begin{proof}
The proof of (1) is straightforward.  
Let $w_{x}=s_{i_{1}}\cdots s_{i_{\ell}}$ be a reduced expression. 
Then $a_{i_{p}}(s_{i_{p}}\cdots s_{i_{\ell}}x)\not=0$ for all $1\le p \le \ell$, 
and hence $d_{i}^{\pm}(x)=d_{\sigma_{x}(i)}^{\pm}(w_{x}x)$. 
\end{proof}

\begin{lem}\label{lem:sigma-ordering}
Suppose that $1\le t_{1}<\cdots <t_{r}\le k$ and  
that $x \in L$ satisfies $\epsilon_{t_{1}}(x)=\cdots =\epsilon_{t_{r}}(x)$. 

(1)\, The values $\sigma_{x}(t_{p})+d_{t_{p}}^{+}(x)$ and 
$\sigma_{x}(t_{p})-d_{t_{p}}^{-}(x)$ are independent of $p=1, 2, \ldots , r$. 

(2)\, Set $l^{\pm}=\sigma_{x}(t_{1})\pm d_{t_{1}}^{\pm}(x)$. 
Then $a_{l^{-}-1}(w_{x}x)>0$, $a_{i}(w_{x}x)=0 \,\, (l^{-}\le i<l^{+}), \, a_{l^{+}}(w_{x}x)=0$ 
and $l^{-}\le \sigma_{x}(t_{1})<\cdots <\sigma_{x}(t_{r}) \le l^{+}$.  
\end{lem}

\begin{proof}
Since $\epsilon_{\sigma_{x}(t_{1})}(w_{x}x)=\cdots=\epsilon_{\sigma_{x}(t_{r})}(w_{x}x)$ 
and $w_{x}x \in \overline{C_{+}}$, 
there exist two integers $l^{\pm}$ such that $1\le l^{\pm}\le k$, 
$a_{l^{-}-1}(w_{x}x)>0$, $a_{i}(w_{x}x)=0 \,\, (l^{-}\le i<l^{+})$, $a_{l^{+}}(w_{x}x)>0$ 
and $l^{-}\le \sigma_{x}(t_{p})\le l^{+}$ for all $1\le p \le r$. 
Then we have  
\begin{align*}
d_{t_{p}}^{\pm}(x)=d_{\sigma_{x}(t_{p})}^{\pm}(w_{x}x)=\pm(l^{\pm}-\sigma_{x}(t_{p})) \quad (1\le p \le r).  
\end{align*}
Therefore $\sigma_{x}(t_{p})\pm d_{t_{p}}^{\pm}(x)$ is equal to $l^{\pm}$ for all $1\le p \le r$. 
{}From the definition of $d_{i}^{+}$ we have $d_{t_{1}}^{+}(x)>\cdots >d_{t_{r}}^{+}(x)$. 
Hence it holds that $\sigma_{x}(t_{1})<\cdots <\sigma_{x}(t_{r})$ 
because $l^{+}=d_{t_{p}}^{+}(x)+\sigma_{x}(t_{p})$ is independent of $p$. 
\end{proof}

The following lemma is the key to the proof of Proposition \ref{prop:G-shift}. 

\begin{lem}\label{lem:group-elt-shift}
Suppose that $1\le t_{1}<\cdots <t_{r}\le k$ and 
that $x \in L$ satisfies $\epsilon_{t_{1}}(x)=\cdots =\epsilon_{t_{r}}(x)$. 
Set $y=x-\sum_{p=1}^{r}v_{t_{p}}$,  
$l=\sigma_{x}(t_{1})+d_{t_{1}}^{+}(x), m=\sigma_{y}(t_{1})-d_{t_{1}}^{-}(y)$ and 
\begin{align}\label{eq:group-elt-shift}
u_{1}&:=(s_{l-r}\cdots s_{\sigma_{x}(t_{1})})(s_{l-r+1}\cdots s_{\sigma_{x}(t_{2})})\cdots 
(s_{l-1}\cdots s_{\sigma_{x}(t_{r})}), \\ 
u_{2}&:=(s_{m+r-1}\cdots s_{\sigma_{y}(t_{r})-1})(s_{m+r-2}\cdots s_{\sigma_{y}(t_{r-1})-1})\cdots 
(s_{m}\cdots s_{\sigma_{y}(t_{1})-1}). 
\nonumber 
\end{align}
Then $u_{1}w_{x}=u_{2}w_{y}$ and 
$\ell(u_{1}w_{x})=\ell(u_{2}w_{y})=\ell(u_{1})+\ell(w_{x})=\ell(u_{2})+\ell(w_{y})$.  
(Note that the right hand sides of \eqref{eq:group-elt-shift} 
are reduced expressions of $u_{1}$ and $u_{2}$.)
\end{lem}

\begin{proof}
Set $z=x-\sum_{p=1}^{r}v_{t_{p}}/2$. 
Since $|\sum_{p=1}^{r}a(v_{t_{p}})|\le 1$ for any $a \in R^{+}$, 
$I(x)$ and $I(y)$ are contained in $I(z)$. 
Hence Proposition \ref{prop:key} implies that $w_{z}=w_{w_{x}z}w_{x}=w_{w_{y}z}w_{y}$ 
and $\ell(w_{w_{x}z})+\ell(w_{x})=\ell(w_{w_{y}z})+\ell(w_{y})$. 
Thus it suffices to show that $w_{w_{x}z}=u_{1}$ and $w_{w_{y}z}=u_{2}$. 
Here we prove that $w_{w_{x}z}=u_{1}$. 
The proof for $w_{w_{y}z}=u_{2}$ is similar. 

Let us write down the set $I(w_{x}z)$. 
It consists of all the positive roots $a \in R^{+}$ such that 
$a(w_{x}z)=a(w_{x}x)-\sum_{p=1}^{r}a(v_{\sigma_{x}(t_{p})})/2<0$. 
Since $w_{x}x \in \overline{C_{+}}$ it is equivalent to requiring that 
$a(w_{x}x)=0$ and there exists $p$ such that 
$a(v_{\sigma_{x}(t_{p})})=1$ and $a(v_{\sigma_{x}(t_{q})})=0$ if $q\not=p$. 
Therefore, from Lemma \ref{lem:sigma-ordering}, we find that 
\begin{align*}
I(w_{x}z)=\bigsqcup_{p=1}^{r}\{ 
\alpha_{\sigma_{x}(p), i} \, | \, \sigma_{x}(t_{p})<i\le l, \, 
i\not=\sigma_{x}(t_{p+1}), \ldots , \sigma_{x}(t_{r})
\}.  
\end{align*}
It is equal to $R^{+}\cap u_{1}^{-1}R^{-}$, and hence $w_{w_{x}z}=u_{1}$. 
\end{proof}

Now let us prove Proposition \ref{prop:G-shift}. 
We use the notation of Lemma \ref{lem:group-elt-shift}. 
Since $a_{i}(w_{y}y)=0$ for $m \le i <\sigma_{y}(t_{r})$ and 
$a_{i}(w_{x}x)=0$ for $\sigma_{x}(t_{1}) \le i <l$, it holds that 
\begin{align*}
w_{y}y=u_{2}w_{y}y=u_{1}w_{x}(x-\sum_{p=1}^{r}v_{t_{p}})=
u_{1}(w_{x}x-\sum_{p=1}^{r}v_{\sigma_{x}(t_{p})})=w_{x}x-\!\!\sum_{j=l-r+1}^{l}\!\!v_{j} 
\end{align*}
and $(\widehat{T}_{u_{2}}^{-1}g)(w_{y}y)=g(w_{y}y)$ for any $g \in F(L)$.
Moreover 
$\widehat{T}_{w_{y}}=\widehat{T}_{u_{2}}^{-1}\widehat{T}_{u_{1}}\widehat{T}_{w_{x}}$. 
Therefore  
\begin{align*}
(\widehat{X}_{t_{1}}\cdots \widehat{X}_{t_{r}}G(f))(x)&=(\widehat{T}_{w_{y}}f)(w_{y}y)=
(\widehat{T}_{u_{1}}\widehat{T}_{w_{x}}f)(w_{x}x-\!\!\sum_{j=l-r+1}^{l}\!\!v_{j}) \\ 
&=(\widehat{X}_{l-r+1}\cdots \widehat{X}_{l}\widehat{T}_{u_{1}}\widehat{T}_{w_{x}}(f))(w_{x}x).
\end{align*}
This completes the proof of Proposition \ref{prop:G-shift}.  

\subsection{Commutation relation of $H$ and $G$}

In this subsection we prove the following theorem. 

\begin{thm}\label{thm:main}
It holds that $HG=G(\sum_{i=1}^{k}\widehat{X}_{i})$. 
Therefore, if $f \in F(L)$ is an eigenfunction of the difference operator $\sum_{i=1}^{k}\widehat{X}_{i}$, 
then $G(f)$ is an eigenfunction of the discrete Hamiltonian $H$ with the same eigenvalue.   
\end{thm}

Hereafter we set 
\begin{align*}
\widehat{V}_{i}:=\alpha+\beta \widehat{X}_{i} \qquad (1\le i \le k).  
\end{align*}

\begin{lem}\label{lem:T-turns-to-identity}
Suppose that $1\le i<k$ and that $x \in L$ satisfies $a_{i}(x)=0$. 
Then for any $g \in F(L)$ and $p \ge 1$ it holds that 
\begin{align*}
\left((q^{p-1}+\delta[p-1]\widehat{V}_{i})(\widehat{T}_{i}+\beta(\gamma+\delta \widehat{X}_{i+1}))(g)
\right)(x)=\left(
(q^{p}+\delta[p]\widehat{V}_{i+1})(g)
\right)(x).  
\end{align*}
\end{lem}

\begin{proof}
For simplicity we write $P_{1} \equiv P_{2}$ 
if two operators $P_{1}, P_{2}$ acting on $F(L)$ satisfy $(P_{1}(g))(x)=(P_{2}(g))(x)$. 
Since $a_{i}(x)=0$ it holds that $\widehat{T}_{i}Q\equiv Q$ for any operator $Q$ acting on $F(L)$. 
Using $\beta(\gamma+\delta\widehat{X}_{i+1})=q-1+\delta\widehat{V}_{i+1}$ we obtain 
\begin{align*}
& 
(q^{p-1}+\delta[p-1]\widehat{V}_{i})(\widehat{T}_{i}+\beta(\gamma+\delta \widehat{X}_{i+1})) \\ 
&\equiv 
q^{p-1}(q+\delta\widehat{V}_{i+1})+\delta[p-1]\widehat{V}_{i}\widehat{T}_{i}+
\delta[p-1]\widehat{V}_{i}(q-1+\delta\widehat{V}_{i+1}). 
\end{align*}
Since $\widehat{V}_{i}\widehat{T}_{i}=
\widehat{T}_{i}\widehat{V}_{i+1}-\widehat{V}_{i}(q-1+\delta\widehat{V}_{i+1})$,  
it is equivalent to 
\begin{align*}
q^{p-1}(q+\delta\widehat{V}_{i+1})+\delta[p-1]\widehat{V}_{i+1}=q^{p}+\delta[p]\widehat{V}_{i+1}. 
\end{align*}
This completes the proof. 
\end{proof}

\begin{lem}\label{lem:Sa}
Suppose that $r, l$ and $\nu_{1}, \ldots , \nu_{r}$ are positive integers satisfying 
$1\le \nu_{1}<\cdots <\nu_{r}\le l\le k$. 
For a subset $I=\{i_{1}, \ldots , i_{d}\} \, (i_{1}<\cdots <i_{d})$ of 
$\{1, 2, \ldots , r\}$, set 
\begin{align*}
c(I):=(\beta/q \alpha)^{d}q^{\sum_{p=1}^{d}i_{p}}, \quad 
\widehat{Q}_{I}:=\widehat{X}_{l-d+1}\cdots \widehat{X}_{l}
(\widehat{T}_{l-d}\cdots \widehat{T}_{\nu_{i_{1}}})\cdots 
(\widehat{T}_{l-1}\cdots \widehat{T}_{\nu_{i_{d}}}). 
\end{align*}
For $1\le p \le r$ define the operator $\widehat{S}_{p}: F(L) \to F(L)$ by 
\begin{align}\label{eq:Sa}
\widehat{S}_{p}:=\sum_{I \subset \{p, p+1, \ldots , r\}}\alpha^{r-p+1}c(I)
\prod_{i=\nu_{p}}^{l-|I|}(q^{p-1}+\delta[p-1]\widehat{V}_{i})\cdot \widehat{Q}_{I}.   
\end{align}
If $x \in L$ satisfies $a_{i}(x)=0$ for $\nu_{1}\le i<l$, 
then it holds that 
\begin{align*}
(\widehat{S}_{p}(g))(x)=(1+\beta\gamma[p-1])
\left(\widehat{S}_{p+1}
\widehat{V}_{\nu_{p}}\!\!\!\prod_{\nu_{p}<i<\nu_{p+1}}\!\!\!(q^{p}+\delta[p]\widehat{V}_{i})\,(g)
\right)(x) 
\end{align*}
for $1\le p \le r$ and $g \in F(L)$, where $\nu_{r+1}=l+1$ and $\widehat{S}_{r+1}$ is the identity operator. 
\end{lem}

\begin{proof}
We use the same symbol $\equiv$ defined in the proof of Lemma \ref{lem:T-turns-to-identity}. 
Decompose the sum $(\widehat{S}_{p}(g))(x)$ into the two parts with $p\in I$ and $p\not\in I$. 
We set $J=I\setminus\{p\}$ and $J=I$ for the first part and the second, respectively. 
Each term in the sum of the second part is invariant 
under the action of $\widehat{T}_{l-|J|-1}\cdots \widehat{T}_{\nu_{p}}$ 
because $a_{i}(x)=0$ for $\nu_{p}\le i<l-|J|$. 
Since $\widehat{T}_{i} \, (\nu_{p}\le i <l-|J|)$ commutes with 
any symmetric polynomial in $\widehat{V}_{i} \, (\nu_{p}\le i \le l-|J|)$, we have 
\begin{align*}
\widehat{S}_{p}&\equiv\sum_{J \subset \{p+1, \ldots , r\}}
\alpha^{r-p}c(J) \prod_{i=\nu_{p}}^{l-|J|-1}(q^{p-1}+\delta[p-1]\widehat{V}_{i})\\ 
&\qquad \qquad {}\times
\left\{\beta q^{p-1}\widehat{X}_{l-|J|}+\alpha(q^{p-1}+\delta[p-1]\widehat{V}_{l-|J|})\right\}
(\widehat{T}_{l-|J|-1}\cdots \widehat{T}_{\nu_{p}})
\widehat{Q}_{J}.   
\end{align*} 
It holds that 
\begin{align*}
\beta q^{p-1}\widehat{X}_{l-|J|}+\alpha(q^{p-1}+\delta[p-1]\widehat{V}_{l-|J|})=
(1+\beta\gamma[p-1])\widehat{V}_{l-|J|}
\end{align*}
and 
\begin{align*}
\widehat{V}_{l-|J|}\widehat{T}_{l-|J|-1}\cdots \widehat{T}_{\nu_{p}}=
\prod_{\nu_{p}\le i <l-|J|}^{\curvearrowleft}
(\widehat{T}_{i}+\beta(\gamma+\delta \widehat{X}_{i+1}))\cdot 
\widehat{V}_{\nu_{p}},
\end{align*}
where $\prod_{m\le i<m'}^{\curvearrowleft}A_{i}:=A_{m'-1}A_{m'-2}\cdots A_{m}$ is an ordered product. 
Now use Lemma \ref{lem:T-turns-to-identity} repeatedly. 
Since $\prod_{\nu_{p}<i<\nu_{p+1}}(q^{p}+\delta[p]\widehat{V}_{i})$ and $\widehat{V}_{\nu_{p}}$ commute 
with $\widehat{Q}_{J}$ for any $J \subset \{p+1, \ldots , r\}$, 
we obtain the desired equality. 
\end{proof}

\begin{prop}\label{prop:Xproduct-to-V}
Suppose that $1\le t_{1}<\cdots <t_{r}\le k$ and 
that $x \in L$ satisfies $\epsilon_{t_{1}}(x)=\cdots =\epsilon_{t_{r}}(x)$. 
Then for any $f \in F(L)$ it holds that 
\begin{align*}
& 
\left(\prod_{p=1}^{r}(\alpha+\beta q^{p-1}\widehat{X}_{t_{p}})G(f)\right)(x) \\ 
&=\prod_{p=1}^{r}(1+\beta\gamma[p-1])
\left\{\prod_{p=1}^{r}\left(\widehat{V}_{\sigma_{x}(t_{p})}
\!\!\!\prod_{\sigma_{x}(t_{p})<i<\sigma_{x}(t_{p+1})}\!\!\!
(q^{p}+\delta[p]\widehat{V}_{i})\right)\widehat{T}_{w_{x}}(f)
\right\}(w_{x}x), 
\end{align*}
where $\sigma_{x}(t_{r+1})=\sigma_{x}(t_{1})+d_{t_{1}}^{+}(x)+1$. 
\end{prop}

\begin{proof}
Proposition \ref{prop:G-shift} implies that 
the left hand side is equal to $(\widehat{S}_{1}\widehat{T}_{w_{x}}(f))(w_{x}x)$, where 
$\widehat{S}_{1}$ is the operator defined by \eqref{eq:Sa} with 
$\nu_{p}=\sigma_{x}(t_{p}) \, (1 \le p \le r)$ and $l=\sigma_{x}(j_{1})+d_{j_{1}}^{+}(x)$.  
Since we have  
$a_{i}(w_{x}x)=0$ for $\sigma_{x}(t_{1})\le i<\sigma_{x}(t_{1})+d_{t_{1}}^{+}(x)$ 
because of Lemma \ref{lem:sigma-ordering}, 
we can apply Lemma \ref{lem:Sa} repeatedly and get the above formula. 
\end{proof}

Now we are ready to prove Theorem \ref{thm:main}. 
Fix $x \in L$ and decompose $\{1, 2, \ldots , k\}=\sqcup_{n=1}^{N}J_{n}^{x}$ 
as described in Lemma \ref{lem:H-decompose}. 
Then each set $\sigma_{x}(J_{n}^{x})$ is an interval of successive integers. 
Take one component $J_{n}^{x}$ and 
set $l^{-}=\min{\sigma_{x}(J_{n}^{x})}$ and $l^{+}=\max{\sigma_{x}(J_{n}^{x})}$. 
{}From Proposition \ref{prop:H-rewrite} and Proposition \ref{prop:Xproduct-to-V}, we see that 
\begin{align*}
& 
(H_{J_{n}^{x}}G(f))(x)=-\frac{\alpha}{\beta}m\,\widehat{T}_{w_{x}}(f)(w_{x}x) \\ 
&+\frac{1}{\beta}
\sum_{r=1}^{m}(-\delta)^{r-1}[r-1]!  \\ 
& \qquad {}\times 
\sum_{l^{-} \le \nu_{1}<\cdots <\nu_{r}\le l^{+}}\!\!\!q^{\sum_{p=1}^{m}(\nu_{p}-l^{+})}
\left(
\prod_{p=1}^{r}(\widehat{V}_{\nu_{p}}\!\!\!
\prod_{\nu_{p}<i<\nu_{p+1}}(q^{p}+\delta[p]\widehat{V}_{i}))\cdot\widehat{T}_{w_{x}}(f)\right)
(w_{x}x), 
\end{align*}
where $\nu_{r+1}=l^{+}+1$.  
Now use the polynomial identity 
\begin{align*}
\sum_{r=1}^{m}(-\delta)^{r-1}[r-1]!\sum_{1\le c_{1}<\cdots <c_{r}\le m}
q^{\sum_{p=1}^{r}(c_{p}-m)}\prod_{p=1}^{r}\left( z_{c_{a}}
\!\!\!\prod_{c_{a}<i<c_{a+1}}\!\!\!(q^{p}+\delta [p] z_{i})
\right)=\sum_{i=1}^{m}z_{i},  
\end{align*} 
where $z_{1}, \ldots , z_{m}$ are commutative indeterminates and $c_{r+1}=m+1$. 
Finally we find that 
\begin{align*}
\left(H_{J_{n}^{x}}G(f)\right)(x)=
(\sum_{j \in J_{n}^{x}}\widehat{X}_{\sigma_{x}(j)}\widehat{T}_{w_{x}}(f))(w_{x}x).   
\end{align*} 
{}From \eqref{eq:H-to-HJ} we have 
\begin{align*}
(HG(f))(x)=(\sum_{n=1}^{N}\sum_{j \in J_{n}^{x}}\widehat{X}_{\sigma_{x}(j)}\widehat{T}_{w_{x}}(f))(w_{x}x)=
(\sum_{j=1}^{k}\widehat{X}_{j}\widehat{T}_{w_{x}}(f))(w_{x}x).  
\end{align*}
Since $\sum_{j=1}^{k}\widehat{X}_{j}$ commutes with $\widehat{T}_{i} \, (1\le i<k)$, 
it is equal to 
\begin{align*}
(\widehat{T}_{w_{x}}\sum_{j=1}^{k}\widehat{X}_{j}(f))(w_{x}x)=G(\sum_{j=1}^{k}\widehat{X}_{j}f)(x). 
\end{align*}
This completes the proof of Theorem \ref{thm:main}.

\subsection{Bethe wave functions}

Using Theorem \ref{thm:main} we can construct symmetric eigenfunctions of $H$, 
which we call the \textit{Bethe wave functions}. 
Set 
\begin{align*}
L_{+}:=L\cap \overline{C_{+}}.  
\end{align*} 

\begin{prop}\label{prop:Bethe-function}
For a tuple $p=(p_{1}, \ldots , p_{k})$ of distinct complex parameters, define the function 
$\Phi_{p} \in F(L)^{W}$ by 
\begin{align}\label{eq:Bethe-function}
\Phi_{p}|_{L_{+}}=\sum_{\sigma \in \mathfrak{S}_{k}}\prod_{1\le i<j \le k}
\left(1+\frac{(\alpha+\beta p_{\sigma(j)})(\gamma+\delta p_{\sigma(i)})}{p_{\sigma(j)}-p_{\sigma(i)}}\right)
\prod_{i=1}^{k}p_{\sigma(i)}^{-\epsilon_{i}}. 
\end{align}
Then $\Phi_{p}$ is an eigenfunction of the discrete Hamiltonian $H$ with eigenvalue $\sum_{i=1}^{k}p_{i}$. 
\end{prop}

\begin{proof}
Denote by $h_{p}$ the function defined by the right hand side of \eqref{eq:Bethe-function} 
on the whole lattice $L$. 
For $\lambda \in V^{*}$ define the function $e^{\lambda} \in F(L)$ by $e^{\lambda}(x):=e^{\lambda(x)}$. 
Then it holds that 
\begin{align*}
\widehat{T}_{i}e^{\lambda}=\left( 
s_{i}+\frac{(\alpha+\beta e^{\lambda(v_{i})})(\gamma+\delta e^{\lambda(v_{i+1})})}
{e^{\lambda(v_{i})}-e^{\lambda(v_{i+1})}}(s_{i}-1)
\right)e^{\lambda}
\end{align*} 
for $1\le i<k$. 
It implies that $\widehat{T}_{i}h_{p}=h_{p}$ for any $1\le i<k$. 
Hence we have 
\begin{align*}
G(h_{p})(x)=(\widehat{T}_{w_{x}} h_{p})(w_{x}x)=h_{p}(w_{x}x)=\Phi_{p}(w_{x}x)=\Phi_{p}(x),  
\end{align*}
that is $G(h_{p})=\Phi_{p}$. 
Since $h_{p}$ is an eigenfunction of $\sum_{j=1}^{k}\widehat{X}_{j}$ with eigenvalue $\sum_{j=1}^{k}p_{j}$, 
it holds that $H\Phi_{p}=(\sum_{j=1}^{k}p_{j})\Phi_{p}$ because of Theorem \ref{thm:main}. 
\end{proof}

%%%%%%%%%%%%%%%%%%%%%%%%%%%%%%%%%%%%%%%%%%%%%%%%%

\section{Construction of Integrable Stochastic Particle System}\label{sec:4}

Hereafter we identify the space of symmetric functions $F(L)^{W}$ with 
the space of functions on $L_{+}$. 
Denote it by $F(L_{+})$.  
A linear operator $Q$ on $F(L_{+})$ is said to be \textit{stochastic} if 
it is given in the form 
$(Qf)(x)=\sum_{y\not=x}c(y, x)(f(y)-f(x))$ where $c(y, x)\ge 0$. 

A stochastic operator on $F(L_{+})$ determines 
a stochastic one-dimensional particle system with continuous time as follows.  
Denote by $S_{k}$ the set of configurations of $k$ bosonic particles 
on the one-dimensional lattice $\mathbb{Z}$. 
For $x=\sum_{j=1}^{k}m_{j}v_{j}$, 
denote by $\nu(x)$ the configuration of $k$ particles on $\mathbb{Z}$ such that 
the particles are on the sites  $m_{1}, \ldots , m_{k}$. 
For example, if $k=6$ and $x=3v_{1}+3v_{2}+3v_{3}+v_{4}-2v_{5}-2v_{6}$, 
$\nu(x)$ is the configuration where  
three particles are located on the site $3$, one particle on the site $1$ 
and two particles on the site $-2$. 
Then the map $\nu: L_{+} \to S_{k}$ is bijection. 
We identify $F(L_{+})$ and the set of functions on $S_{k}$ through the map $\nu$. 
Then the stochastic operator $Q$ 
is regarded as the backward generator of the stochastic process on $S_{k}$ with continuous time,   
where $c(y, x)$ gives the rate at which the state $\nu(x)$ changes to $\nu(y)$. 

Now we give a sufficient condition for $H|_{F(L)^{W}}=H|_{F(L_{+})}$ to be 
stochastic up to constant. 

\begin{prop}\label{prop:stochastic}
Let $\lambda$ be a constant. 
The operator $\tilde{H}:=(H+\lambda)|_{F(L_{+})}$ 
is stochastic only if $\lambda=-k$ and $(\alpha+\beta)(\gamma+\delta)=0$.  
\end{prop}

To prove Proposition \ref{prop:stochastic}, 
we introduce the cluster coordinate 
of a point in $L_{+}$ following \cite{BCPS}.  
For $x \in L_{+}$ we determine a set of positive integers 
$M$ and $c_{i} \, (1 \le i \le M)$ by the property that 
$\sum_{i=1}^{M}c_{i}=k$, 
$\epsilon_{c_{1}}(x)>\epsilon_{c_{1}+c_{2}}(x)>\cdots >\epsilon_{c_{1}+\cdots +c_{M}}(x)$, 
and 
$\epsilon_{j}(x)=\epsilon_{c_{1}+\cdots +c_{i}}(x)$ 
if $c_{1}+\cdots +c_{i-1}<j\le c_{1}+\cdots +c_{i}$.  
We call the tuple $(c_{1}, \ldots , c_{M})$ the \textit{cluster coordinate} of $x \in L_{+}$. 
It describes the number of particles in each cluster  
in the configuration $\nu(x)$. 

In terms of the cluster coordinate 
the action of $H$ for $f \in F(L)^{W}$ is written as follows. 
Fix $x \in L_{+}$ and let $(c_{1}, \ldots , c_{M})$ be its cluster coordinate. 
Then 
\begin{align}\label{eq:Htilde-decompose}
& 
(Hf)(x)=\sum_{i=1}^{M}\Bigl\{ -\alpha\gamma\sum_{d=1}^{c_{i}-1}\frac{[d]}{1+\beta\gamma[d]}f(x) 
\\ &+\sum_{r=1}^{c_{i}}
\frac{(-\beta\delta)^{r-1}[r-1]!q^{-r(r-1)}}{\prod_{p=0}^{r-1}(1+\beta\gamma[c_{i}-1-p])}
e_{r}(1, q, \ldots , q^{c_{i}-1})f(x-\sum_{p=0}^{r-1}v_{c_{1}+\cdots +c_{i}-p})
\Bigr\}. 
\nonumber 
\end{align}
We use this formula in the proof below. 

\begin{proof}[Proof of Proposition \ref{prop:stochastic}]
If $c_{1}, \ldots , c_{M}$ are all equal to one, then 
then $(Hf)(x)=\sum_{j=1}^{k}f(x-v_{j})$. 
Hence $\lambda$ should be equal to $-k$ so that $\tilde{H}$ is stochastic. 

In general, set 
\begin{align}\label{eq:Htilde-decompose-coeff}
K_{m}&:=-m-\alpha\gamma\sum_{d=1}^{m-1}\frac{[d]}{1+\beta\gamma[d]} \\ 
&+\sum_{r=1}^{m}\frac{(-\beta\delta)^{r-1}[r-1]!\,q^{-r(r-1)}}{\prod_{p=0}^{r-1}(1+\beta\gamma[m-1-p])}
e_{r}(1, q, \ldots , q^{m-1}). 
\nonumber 
\end{align}
The operator $\tilde{H}$ is stochastic only if $\sum_{i=1}^{M}K_{c_{i}}=0$ 
for any tuple $(c_{1}, \ldots , c_{M})$ of positive integers such that $\sum_{i=1}^{M}c_{i}=k$. 
Since $K_{1}=0$ and $K_{2}=-(\alpha+\beta)(\gamma+\delta)/(1+\beta\gamma)$, 
we see that $(\alpha+\beta)(\gamma+\delta)$ should be zero. 
\end{proof}

Moreover, we have the following property. 

\begin{lem}
If $(\alpha+\beta)(\gamma+\delta)=0$, 
the constant $K_{m}$ defined by \eqref{eq:Htilde-decompose-coeff} is equal to zero for any $m\ge 1$. 
\end{lem}

\begin{proof}
{}From Lemma \ref{lem:H-rewrite-1} and Lemma \ref{lem:H-rewrite-2} it holds that 
\begin{align*}
K_{m}&=-(1+\frac{\beta}{\alpha})m \\ 
&+\frac{1}{\beta}
\sum_{r=1}^{m}(-\delta)^{r-1}[r-1]!\,q^{-(m-1)r}e_{r}(1, q, \ldots , q^{m-1})
\prod_{p=1}^{r}\frac{\alpha+\beta q^{p-1}}{1+\beta\gamma[p-1]}.  
\end{align*}
Using this expression we find that 
\begin{align*}
& 
K_{m}-K_{m-1} \\ 
&=-\frac{(\alpha+\beta)(\gamma+\delta)}{1+\beta\gamma}
\sum_{r=1}^{m-1}(-\delta)^{r-1}[r]!\,q^{-(m-2)r}e_{r}(1, q, \ldots , q^{m-2})
\prod_{p=2}^{r}\frac{\alpha+\beta q^{p-1}}{1+\beta\gamma[p]}  
\end{align*}
for $m\ge 2$. 
This completes the proof because $K_{1}=0$. 
\end{proof}

Now we define the stochastic operator $\mathcal{H}(s, q)$ on $F(L_{+})$ by 
\begin{align*}
(\mathcal{H}(s, q)f)(x)=\sum_{i=1}^{M}\sum_{r=1}^{c_{i}}
\frac{s^{r-1}}{[r]}\prod_{p=0}^{r-1}\frac{[c_{i}-p]}{1+s[c_{i}-1-p]}
\left( 
f(x-\sum_{p=0}^{r-1}v_{c_{1}+\cdots +c_{i}-p})-f(x) 
\right), 
\end{align*}
where $(c_{1}, \ldots , c_{M})$ is the cluster coordinate of $x$. 
It determines the stochastic particle system on $\mathbb{Z}$ 
described as follows. 
In continuous time some particles may move from site $i$ to $i-1$ 
independently for each $i \in \mathbb{Z}$. 
The rate at which $r$ particles move to the left from a cluster with $c$ particles 
is given by 
\begin{align*}
\frac{s^{r-1}}{[r]}\prod_{p=0}^{r-1}\frac{[c-p]}{1+s[c-1-p]} \qquad (c \ge r \ge 1).  
\end{align*}
It is non-negative if, for example, $s\ge 0$ and $0<q<1$. 

As a consequence we have the following proposition. 

\begin{prop}
When $(\alpha+\beta)(\gamma+\delta)=0$, it holds that 
\begin{align*}
(H-k)|_{F(L_{+})}=\left\{
\begin{array}{ll}
\mathcal{H}(q^{-1}\alpha\delta, q^{-1}) & (\alpha+\beta=0), \\ 
\mathcal{H}(\beta\gamma, q) & (\gamma+\delta=0). 
\end{array}
\right. 
\end{align*}
\end{prop}

\begin{proof}
Use the equality 
\begin{align*}
q^{-r(r-1)/2}e_{r}(1, q, \ldots , q^{m-1})=\frac{\prod_{p=0}^{r-1}[m-p]}{[r]!},  
\end{align*} 
and we obtain the desired formula. 
\end{proof}

Moreover, using Proposition \ref{prop:Bethe-function}, we obtain 
eigenfunctions of $\mathcal{H}(s, q)$:  

\begin{prop}
Let $z=(z_{1}, \ldots , z_{k})$ be a tuple of distinct complex parameters, and set 
\begin{align}\label{eq:def-nu}
\nu:=\frac{s}{1-q+s}.
\end{align}
Then the function $\Psi_{z}$ on $L_{+}$ defined by 
\begin{align*}
\Psi_{z}:=\sum_{\sigma \in \mathfrak{S}_{k}}\prod_{1\le i<j \le k}
\frac{qz_{\sigma(i)}-z_{\sigma(j)}}{z_{\sigma(i)}-z_{\sigma(j)}}
\prod_{i=1}^{k}\left(\frac{1-\nu z_{\sigma(i)}}{1-z_{\sigma(j)}}\right)^{\epsilon_{i}}  
\end{align*}
satisfies 
\begin{align*}
\mathcal{H}(s, q)\Psi_{z}=(\nu-1)\sum_{i=1}^{k}\frac{z_{i}}{1-\nu z_{i}}\Psi_{z}.  
\end{align*}
\end{prop}

\begin{proof}
We use Proposition \ref{prop:Bethe-function} in the case where $\delta=-\gamma$. 
Note that $q=1+\beta\gamma-\alpha\delta=1+(\alpha+\beta)\gamma$. 
Setting $p_{i}=(1-z_{i})/(1+\beta z_{i}/\alpha)$, we have 
\begin{align*}
1+\frac{(\alpha+\beta p_{j})(\gamma+\delta p_{i})}{p_{j}-p_{i}}=\frac{qz_{i}-z_{j}}{z_{i}-z_{j}}.    
\end{align*}
Set $s=\beta\gamma$. Then $\beta/\alpha$ is equal to $-\nu$. %R defined by \eqref{eq:def-nu}.  
Thus we see that $\Psi_{z}$ is an eigenfunction of 
$\mathcal{H}(s, q)=(H-k)|_{F(L_{+})}$ with eigenvalue
\begin{align*}
\sum_{i=1}^{k}p_{i}-k=\sum_{i=1}^{k}\frac{1-z_{i}}{1-\nu z_{i}}-k=
(\nu-1)\sum_{i=1}^{k}\frac{z_{i}}{1-\nu z_{i}}.  
\end{align*} 
This completes the proof. 
\end{proof}

It should be noted that the function $\Psi_{z}$ is equal to 
the eigenfunction for the $(q, \mu, \nu)$-Boson process 
constructed by means of the coordinate Bethe ansatz \cite{P}.

%%%%%%%%%%%%%%%%%%%%%%%%%%%%%%%%%%%%%%%%%%%%%%%%%

\appendix

\section{}\label{sec:app}

Here we prove Lemma \ref{lem:H-rewrite-1} and Lemma \ref{lem:H-rewrite-2}. 
For that purpose we show the following equality. 

\begin{lem}\label{lem:Ims}
Let $m$ be a positive integer and $x, y, z_{1}, \ldots , z_{m}$ commutative indeterminates. 
For $1\le s \le m$, set 
\begin{align}\label{eq:Ims}
I_{m, s}(x, y)&:=
\sum_{r=0}^{m-s}\frac{[r+s-1]!\,((q-1)x-y)^{r}}{\prod_{a=1}^{r+s}(x+[a-1]y)}\\ 
&\quad{}\times  
\sum_{1\le b_{1}<\cdots <b_{r+s}\le m}
q^{\sum_{a=1}^{r+s}(b_{a}-m)}
e_{s}(z_{b_{1}}, qz_{b_{2}}, \ldots , q^{r+s-1}z_{b_{r+s}}),  
\nonumber 
\end{align}
where $e_{s}$ is the elementary symmetric polynomial of degree $s$. 
Then it holds that 
\begin{align}\label{eq:Ims-to-show}
I_{m, s}(x, y)=\frac{[s-1]!\,q^{-s(s-1)/2}}{\prod_{a=1}^{s-1}(x+[m-1-a]y)}\,
e_{s}(z_{1}, qz_{2}, \ldots , q^{m-1}z_{m}). 
\end{align} 
\end{lem}

\begin{proof}
First we prove 
\begin{align}\label{eq:Ims-1}
& 
\sum_{1\le b_{1}<\cdots <b_{r+s}\le m}q^{\sum_{a=1}^{r+s}b_{a}}
e_{s}(z_{b_{1}}, qz_{b_{2}}, \ldots , q^{r+s-1}z_{b_{r+s}}) \\ 
&=q^{s(s-1)/2}e_{r}(q^{s+1}, q^{s+2}, \ldots , q^{m})
e_{s}(qz_{1}, q^{2}z_{2}, \ldots , q^{m}z_{m}) 
\nonumber 
\end{align} 
for $m\ge 1, r\ge 0$ and $s\ge 0$ satisfying $r+s\le m$ by induction on $m$. 
If $m=1$ it is trivial. 
Suppose that $m>1$. 
Since the equality holds trivially when $r=0$ or $s=0$, 
we assume that $r>0$ and $s>0$. 
Denote the left hand side by $K_{m, r, s}$. 
Using 
\begin{align*}
e_{s}(z_{b_{1}}, qz_{b_{2}}, \ldots , q^{r+s-1}z_{b_{r+s}})
&=e_{s}(z_{b_{1}}, qz_{b_{2}}, \ldots , q^{r+s-2}z_{b_{r+s-1}}) \\ 
&+q^{r+s-1}z_{b_{r+s}}
e_{s-1}(z_{b_{1}}, qz_{b_{2}}, \ldots , q^{r+s-2}z_{b_{r+s-1}}),  
\end{align*}
we see that 
\begin{align*}
K_{m, r, s}=\sum_{b=r+s}^{m}q^{b}\left(K_{b-1, r-1, s}+q^{r+s-1}z_{b}K_{b-1, r, s-1}\right).  
\end{align*}
{}From the hypothesis of the induction it is equal to 
\begin{align*}
q^{s(s-1)/2}\sum_{b=r+s}^{m}q^{b}\bigl\{
&e_{r-1}(q^{s+1}, \ldots , q^{b-1})e_{s}(qz_{1}, \ldots , q^{b-1}z_{b-1}) \\ 
&+z_{b}\,e_{r}(q^{s+1}, \ldots , q^{b})
e_{s-1}(qz_{1}, \ldots , q^{b-1}z_{b-1})\bigr\}. 
\end{align*}
Use 
\begin{align*}
q^{b}z_{b}\,e_{s-1}(qz_{1}, \ldots , q^{b-1}z_{b-1})=
e_{s}(qz_{1}, \ldots , q^{b}z_{b})-e_{s}(qz_{1}, \ldots , q^{b-1}z_{b-1})
\end{align*}
and 
\begin{align}\label{eq:Ims-2}
q^{b}e_{r-1}(q^{s+1}, \ldots , q^{b-1})-e_{r}(q^{s+1}, \ldots , q^{b})=e_{r}(q^{s+1}, \ldots , q^{b-1}) 
\end{align}
successively. 
Then we get the right hand side of \eqref{eq:Ims-1}. 

Now let us prove \eqref{eq:Ims-to-show}. 
Using \eqref{eq:Ims-1} we see that 
\begin{align*}
I_{m, s}(x, y)=q^{s(s+1)/2-ms}e_{s}(z_{1}, qz_{2}, \ldots , q^{m-1}z_{m})J_{m, s}(x, y),  
\end{align*}
where $J_{m, s}(x, y)$ is given by 
\begin{align*}
J_{m, s}(x, y):=\sum_{r=0}^{m-s}
\frac{[r+s-1]!\,((q-1)x+y)^{r}q^{-mr}}{\prod_{a=1}^{r+s}(x+[a-1]y)}e_{r}(q^{s+1}, q^{s+2}, \ldots , q^{m}).  
\end{align*}
It suffices to show that 
\begin{align}\label{eq:Ims-3}
J_{m, s}(x, y)=q^{-s^{2}+ms}\frac{[s-1]!}{\prod_{a=0}^{s-1}(x+[m-1-a]y)} 
\end{align}
for $1\le s\le m$. 
{}From the equality \eqref{eq:Ims-2} with $b$ replaced by $m$ and  
$x+[n]y=x+y+q[n-1]y$ for $n\ge 1$, we find 
$J_{m, s}(x, y)=q^{s}J_{m-1, s}(x+y, qy)$ for $m>s$.   
Now the desired equality \eqref{eq:Ims-3} can be proved by induction on $m$. 
\end{proof}

\begin{proof}[Proof of Lemma \ref{lem:H-rewrite-1}]
We rewrite the right hand side. 
Expand the product  
\begin{align*}
\prod_{a=1}^{r}(\alpha+\beta q^{a-1}z_{b_{a}})-\alpha^{r}=\sum_{s=1}^{r}
\alpha^{r-s}\beta^{s-1}e_{s}(z_{b_{1}}, qz_{b_{2}}, \ldots , q^{r-1}z_{b_{r}}).    
\end{align*}
and exchange the order of the summation with respect to $r$ and $s$. 
Using  
\begin{align*}
(-\delta)^{r+s-1}\alpha^{r}\beta^{s-1}=(-\beta\delta)^{s-1}(q-1-\beta\gamma)^{r},   
\end{align*}
we see that the right hand side is equal to 
$\sum_{s=1}^{m-1}(-\beta\delta)^{s-1}I_{m, s}(1, \beta\gamma)$, where 
$I_{m, s}(x, y)$ is defined by \eqref{eq:Ims}.  
It is equal to the left hand side because of Lemma \ref{lem:Ims}.  
\end{proof}

\begin{proof}[Proof of Lemma \ref{lem:H-rewrite-2}]
Set 
\begin{align*}
K_{m}(x, y):=\sum_{r=1}^{m}\frac{[r-1]!\,((q-1)x-y)^{r-1}q^{-mr}}{\prod_{a=1}^{r}(x+[a-1]y)}
e_{r}(q, q^{2}, \ldots , q^{m}).  
\end{align*} 
Then the left hand side is equal to $\alpha K_{m}(1, \beta\gamma)/\beta$. 
Hence it suffices to show that 
\begin{align*}
K_{m}(x, y)=\sum_{a=0}^{m-1}\frac{1}{x+[a]y}.  
\end{align*}
In the same way as the proof of \eqref{eq:Ims-3}, 
we find the recurrence relation 
$K_{m}(x, y)=1/x+K_{m-1}(x+y, qy)$ for $m>1$.  
Now the equality above can be proved by induction on $m$. 
\end{proof}

\section*{Acknowledgments}

The research of the author is supported by 
JSPS KAKENHI Grant Number 26400106. 
The author is grateful to I. Corwin, L. Petrov, A. Povolotsky and T. Sasamoto for valuable discussions.

\end{document}